\documentclass{article}
\usepackage{ijcai19}

\usepackage{times}
\usepackage[hidelinks]{hyperref}
\usepackage[utf8]{inputenc}
\usepackage[small]{caption}
\usepackage{soul}
\usepackage{graphicx}
\usepackage{xspace}
\usepackage[numbers, compress, sort, square]{natbib}
\usepackage{subcaption}
\usepackage{tikz}
\usetikzlibrary{shapes,shapes.geometric,arrows,fit,calc,positioning,automata}
\usepackage{amssymb,amsmath,amsthm}

\newtheorem{thm}{Theorem}
\newtheorem{cor}{Corollary}
\newtheorem{pro}{Proposition}
\newtheorem{lemma}{Lemma}
\newtheorem{definition}{Definition}

\newcommand{\s}{{\bf s}}

\newcommand{\sm}[1]{\ensuremath{{\bf s}_{-#1}}}

\renewcommand{\c}{\ensuremath{c_{\textrm{e}}\xspace}}
\newcommand{\CC}{\ensuremath{CC\xspace}}

\newcommand*{\field}[1]{\mathbb{#1}}%
\newcommand{\N}{\field{N}} %natural numbers
\renewcommand{\Pr}{\mathrm{Pr}}%probability
\newcommand{\E}{\mathbb{E}}%expectation

\title{Network Formation under Random Attack and Probabilistic Spread}

\author{
Yu Chen$^1$
\and
Shahin Jabbari$^1$\and
Michael Kearns$^1$\and
Sanjeev Khanna$^1$\And
Jamie Morgenstern$^2$
\affiliations
$^1$University of Pennsylvania\\
$^2$Georgia Tech\\
\emails
\{chenyu2, jabbari, mkearns, sanjeev\}@cis.upenn.edu,
jamiemmt.cs@gatech.edu
}

\begin{document}
\maketitle

\begin{abstract}
We study a network formation game where agents receive 
benefits by forming connections to other agents but also incur
both direct and indirect costs from the formed connections. 
Specifically, once the agents have 
purchased their connections, an attack starts at a randomly 
chosen vertex in the 
network and spreads according to the independent cascade 
model with a fixed probability, 
destroying any infected agents. The utility or welfare of an agent 
in our game is defined to be 
the expected size of the agent's connected component 
post-attack minus her expenditure in 
forming connections. 

Our goal is to understand the properties of the equilibrium 
networks formed in this game. 
Our first result concerns the edge density of equilibrium networks. 
A network connection increases 
both the likelihood of remaining connected to other agents after 
an attack as well the likelihood of 
getting infected by a cascading spread of infection. We show that 
the latter concern primarily prevails and 
any equilibrium network in our game contains only $O(n\log n)$ 
edges where $n$ denotes the number of 
agents. On the other hand, there are equilibrium networks that 
contain $\Omega(n)$ edges showing that our 
edge density bound is tight up to a logarithmic factor. 
Our second result shows that the presence of attack and its 
spread through a cascade does not significantly 
lower social welfare as long as the network is not too dense. 
We show that any non-trivial equilibrium 
network with $O(n)$ edges has $\Theta(n^2)$ social welfare, 
asymptotically similar to the social welfare guarantee
in the game without any attacks. 
\end{abstract}
\section{Introduction}
\label{sec:intro}

We study a network formation game where strategic 
agents (vertices on a graph) 
receive both benefits and costs from forming connections 
to other agents. While 
various benefit functions exist in the 
literature~\cite{BalaG00, FabrikantLMPS03}, 
we focus on the \emph{reachability network benefit}. 
Here, the benefit of an agent is 
the size of her connected component in the collectively 
formed graph. This models 
settings where reachability (rather than centrality) motivates 
joining the network, e.g. 
when transmitting packets over technological networks 
such as the Internet.

Most previous works feature a direct edge cost $\c > 0$ 
for forming a 
link.~\citet{GoyalJKKM16} depart from this notion by 
studying a game 
where forming links introduces an additional 
\emph{indirect} cost by exposing agents 
to contagious network shocks. These indirect 
costs can model scenarios such 
as virus spread through technological or biological 
networks.

Our work continues this investigation of direct and 
indirect connection costs. 
To model the indirect cost we assume that, after 
network formation, an adversary attacks a 
single vertex uniformly at random. The attack then 
kills the vertex and spreads through the 
network via the independent cascade model 
according to parameter $p$~\cite{KempeKT03}. 
This random attack and probabilistic spread captures 
the epidemiological quality of virus spread 
in both biological and technological networks.

At a high level, our work is most closely related 
to two previous works.~\citet{BalaG00} study a 
reachability network game without attacks and show a 
sharp characterization of equilibrium networks: 
every tree and the empty network can form in 
equilibria.~\citet{GoyalJKKM16} study a reachability 
network formation game where an adversary 
inspects the formed network and then deliberately 
attacks a single vertex in the network. The attack 
then spreads deterministically to neighboring vertices 
according to a known rule, while agents may immunize 
against the attack for a fixed cost. Our 
game is most similar to the latter setting under a 
random adversary and high immunization cost. 
However, in our setting attacks spread probabilistically 
(through independent cascades) rather 
than deterministically. This yields an arguably more 
realistic model of infection spread but incurs 
additional complexity: computing the expected 
connectivity benefit of an agent in a given network 
is now \#P-complete~\cite{WangCW12}.

\citet{GoyalJKKM16} show that while more diverse 
equilibrium networks, including ones with multiple cycles, 
can emerge in addition to trees and the empty 
graph, the equilibrium networks with $n$ agents will have 
at most $2n-4$ edges; 
less than twice the number of edges that can form in the 
equilibria of the attack-free game. 
Furthermore, they show that 
the social welfare is at least $n^2-o(n^{5/3})$ in non-trivial 
equilibrium networks.~Asymptotically, this is  the maximum welfare possible which is 
achieved in any nonempty equilibrium of the attack-free game. 
 In the regime where the cost of immunization is high, 
the game of \citet{GoyalJKKM16} only admits disconnected 
and fragmented equilibrium networks due to deterministic 
spread of the attack, and the social welfare of the
resulting networks may be as low as 
$\Theta(n)$.
$\newline$

\noindent\textbf{Our Results and Techniques}
In our game, computing utilities or even verifying 
network equilibrium is computationally hard. 
We circumvent this difficulty by proving structural 
properties for equilibrium networks. First, 
we provide an upper bound on the edge density in equilibria.
\begin{thm}[Statement of Theorem~\ref{thm:density1}]
\label{thm:1}
Any equilibrium network on $n$ vertices has 
$O\left(n\log n/p\right)$ edges.
\end{thm}
For constant $p$ this upper bound is tight up to a 
logarithmic factor. The possibility of over-building 
therefore differentiates our game from those 
of~\citet{BalaG00} and~\citet{GoyalJKKM16}, but the extent 
of over-building is limited.

To prove Theorem~\ref{thm:1}, we first show 
that any equilibrium network 
with more than $\Omega(n\log(n/p))$ edges 
contains an induced subgraph 
with large minimum cut size. We then show that 
if a network has large minimum cut 
size, in \emph{every} attack (with high probability), 
either almost all vertices in the network 
will die or almost all vertices in the network will survive. 
As a result, any vertex in the induced 
subgraph can beneficially deviate by dropping an edge. 
Together, these observations allows us to prove the claimed
edge density bound.

Next, we show that any equilibrium network that is 
nontrivial (i.e. contains at least one edge) 
also contains a large connected component. 
Moreover, as long as the
network is not too dense, it achieves a constant 
approximation to the best welfare 
possible of the attack-free game.
\begin{thm}[Informal Statement of 
Theorems~\ref{lem:largest}~and~\ref{thm:welfare}]
\label{thm:welf}
Any non-trivial equilibrium network over $n$ 
vertices contains a connected component
of size at least $n/3$. Furthermore, if the 
number of edges in the network is 
$O(n/p)$, then the social welfare is $\Omega(n^2)$.
\end{thm}

To prove Theorem~\ref{thm:welf}, we first show 
that any agent in a small connected 
component can increase her connectivity benefits
by purchasing an edge to a larger 
component without significantly increasing her 
attack risk. This implies the existence 
of a large connected component. We then 
use the large component to argue 
that when the equilibrium network is sparse, 
the surviving network post-attack still contains 
a large connected component. This guarantees 
large social welfare.

While~\citet{GoyalJKKM16} show robustness 
of the structural properties of 
the original reachability game 
of~\citet{BalaG00} to a variation with attack,
deterministic spread and the option of 
immunization for players, we show 
robustness in another variant that involves a 
cascading attack but disallows
immunization. However, on the technical front, 
the tools that we use
to prove these robustness results are very 
different from the analysis 
of both of these previous games.

\noindent\textbf{Organization} 
We introduce our model 
and discuss the related work in Section~\ref{sec:model}. 
In Section~\ref{sec:examples} we present examples of
equilibrium  networks of our game. 
Sections~\ref{sec:density}~and~\ref{sec:welfare} are devoted
to the characterization of the edge density and social welfare. 	
We conclude with directions for future work in Section~\ref{sec:future}.
\section{Model}
\label{sec:model}
We start by formalizing our model and borrow most of our 
notation and terminology from~\citet{GoyalJKKM16}. We assume the $n$ vertices of a graph 
(network) correspond to individual players. Each player has the choice 
to purchase edges to other players at a \emph{fixed} cost of $\c>0$ per edge. 
Throughout we assume that $\c$ is a constant independent of $n$.
Furthermore, we use the term \emph{high probability} to refer to probability 
at least $1-o(1/n)$ henceforth.

A (pure) \emph{strategy} $s_i\subseteq [n]$ for player $i$ 
consists of a subset of players to whom player $i$ purchased an edge.
We assume that edge purchases are unilateral i.e. players do not need
approval to purchase an edge to another player
but that the connectivity benefits and risks are bilateral.\footnote{As 
an example of a scenario where the consequences are bilateral 
even though the link formation is unilateral,
consider the spread of a disease in a social network where the links 
are formed as a result of physical proximity of individuals.  
The social benefits and potential risks of a contagious disease are 
bilateral in this case although the link formation as a result of proximity 
is unilateral. We leave the study of the bilateral edge formation 
for future work.}
  
Let $\s = (s_1,\ldots,s_n)$ denote the strategy profile for all the
players.  Fixing $\s$, the set of edges purchased by all the
players induces an undirected graph. We denote a \emph{game}
\emph{graph} as a graph $G$, where $G=(V,E)$ is the undirected
graph induced by the edge purchases of all players.

Fixing a game graph $G$, the adversary selects a \emph{single}
vertex $v\in V$ uniformly at random to start the attack.
The attack kills $v$ and then spreads according to the 
independent cascade model 
with probability $p\in(0,1)$~\cite{KempeKT03}.\footnote{Throughout we assume 
that $p$ is a constant independent of the number of players $n$. We discuss
the regime in which $p$ decreases as the number of players increases
in Section~\ref{sec:edge-small-p}.} 
In the independent cascade model, in the first round, the attack spreads independently 
killing each of the neighbors of the initially attacked vertex $v$ with probability $p$. 
In the next round, the spread continues from all the neighbors of $v$ that were killed in the 
previous round. 
The spread stops when no new vertex was killed in the last round or when all the vertices are killed.

The adversary's attack can be alternatively described as follows. Fixing a game graph $G$, 
let $G[p]$ denote the \emph{random} graph 
obtained by retaining each edge of $G$ independently with probability $p$. 
The adversary picks a vertex $v$ uniformly at random to start the attack. 
The attack kills $v$ and all the vertices in the connected component of 
$G[p]$ that contains $v$.

Let $\CC_i(v)$ denote the \emph{expected} size of the connected component of player $i$ 
post-attack to a vertex $v$ and we define $\CC_v(v)$ to be $0$. Then the expected utility (utility for short) of
player $i$ in strategy profile $\s$ denoted by $u_i(\s)$ is precisely
\[
u_i(\s) =\frac{1}{|V|}\sum_{v\in V}\CC_i\left(v\right)-|s_i| \c.
\]
We refer to the sum of utilities of all the
players playing a strategy profile $\s$ as the \emph{social welfare} of $\s$.

\citet{WangCW12} show that computing the exact spread of the attack in the independent cascade model
is \#P-complete in general. This implies that, given a strategy profile $\s$, computing the expected 
size of the connected component of all vertices (and hence the expected utility of all vertices) 
is \#P-complete. However, an approximation of these quantities can be obtained 
by Monte Carlo simulation.

We model each of the $n$ players as
strategic agents who deterministically choose which edges to purchase.
A strategy profile $\s$ is a
\emph{pure strategy Nash equilibrium} if,
for any player $i$, fixing the behavior of the other players to be
$\sm{i}$, the expected utility for $i$, $u_i(\s)$, cannot strictly increase when playing any strategy $s'_i$
over $s_i$. We focus our attention to pure strategy Nash equilibrium (or equilibrium) in this work.
Since computing the expected utilities in our game is \#P-complete, even verifying
that a strategy profile is an equilibrium is \#P-complete. Hence as our main contribution, 
we prove structural properties for the equilibrium networks regardless of this computational barrier.
\subsection{Related Work}
\label{sec:related}
There are two lines of work closely related to ours. First,~\citet{BalaG00}
study the attack-free version of our game. They show that equilibrium networks are either 
trees or the empty network. Also since there is no attack, the social welfare in nonempty equilibrium 
networks is asymptotically $n^2-o(n^2)$.

Second,~\citet{GoyalJKKM16} study a network formation game where players in addition to having the option of
purchasing edges can also purchase immunization from the attack. Since we do not study the effect of immunization 
purchases in our game, our game corresponds to the regime of parameters in their game
where the cost of immunization is so high that no vertex would purchase immunization in equilibria.
Moreover, they study several different adversarial attack models and our attack model coincides
with their \emph{random attack adversary}.
The main difference between our work and theirs is that they assume the attack spreads
deterministically while we assume the attack spreads according to the independent cascade 
model~\cite{KempeKT03}. In many real world scenarios e.g. the spread of 
contagious disease over the network of people, the spread is \emph{not deterministic}. Hence our 
work can be seen as a first attempt to make the model of \citet{GoyalJKKM16} closer to real world
applications. However, the change in the spread of attack comes with a significant increase in the 
complexity of the game as even computing the utilities of the players in our game is \#P-complete.
While \citet{FriedrichIKLNS17} have shown that  best responses for players can be computed in polynomial 
time under various attack models, the question of whether best response dynamics
converges to an equilibrium network is open in the model of 
\citet{GoyalJKKM16}.

Similar to \citet{GoyalJKKM16} we show that diverse equilibrium networks 
can form in our game. While they show that
all equilibrium networks over $n\ge 4$ players have at most $2n-4$ edges, we show
that the number of edges in any equilibrium network is at most $O(n\log n)$ 
and this bound is tight up to a logarithmic factor.
Furthermore, \citet{GoyalJKKM16} show that the social welfare is asymptotically $n^2-o(n^2)$ 
in non-trivial equilibrium networks. Their definition of non-trivial networks requires
the network to have at least one immunized vertex and one edge.
 In the regime where the cost of immunization is high, 
the game of \citet{GoyalJKKM16} only admits disconnected and fragmented 
equilibrium networks due to the deterministic spread of the attack.
Such networks (even excluding the empty graph) can have social welfare
as low as $\Theta(n)$. 
We show that any low density equilibrium network of our game 
enjoys a social welfare of $\Theta(n^2)$ as long as the network contains at least one edge.

\citet{Kliemann11} introduced a network formation game with 
reachability benefits and an attack on the formed network that destroys exactly one link
with no further spread. Their equilibrium networks are sparse and also admit high social
welfare as removing an edge can create at most two connected components. 
\citet{KliemannSS17} extend this to allow attacks on vertices while focusing 
on swapstable equilibria. 

\citet{BlumeEKKT11} introduce a network formation game with bilateral
edge formation. They assume both edge  and link failures can happen simultaneously 
but independent of the failures so far in the network. These differences make it hard to 
directly compare the two models. They show a tension between optimal and 
stable networks and exploring such properties in depth in our model is an interesting 
direction.

Finally, network formation games, with a variety of different connectivity benefit models, 
have been studied extensively in computer science see e.g.~\cite{BalaG00, BlumeEKKT11, Kliemann11}. 
We refer the reader to the related work section of~\citet{GoyalJKKM16} for a comprehensive
summary of other related work especially on the topic of optimal security choices for networks.
\section{Examples of Equilibrium Networks}
\label{sec:examples}
In this section we show that a diverse set of topologies can emerge in the equilibrium 
of our game. Similar to the models of~\citet{BalaG00}~and~\citet{GoyalJKKM16} the empty graph can 
form in the equilibrium of our game when $\c\geq1$. 
Moreover, similar to both models, trees can form in equilibria (See the left panel of Figure~\ref{fig:eq-examples}).
Finally, while~\citet{GoyalJKKM16} show that in the regime 
of their game where the cost of immunization is high (so no vertex would immunize)
no connected network can form in equilibria due to the deterministic spread of the attack, 
we show that connected networks indeed can form in the equilibria of our game
(See Figure~\ref{fig:eq-examples}).

\begin{figure}[ht!]
\centering
\begin{subfigure}[b]{0.14\textwidth}
\centering
\begin{tikzpicture}
[scale=0.5, every node/.style={circle,fill=white, draw=black}, gray node/.style = {circle, fill = blue, draw}]
\node (1) at  (0, 0){};
\node (2) at  (-2, 0){};\node (3) at  (-1.41, 1.41){};\node (4) at  (0, 2){};
\node (5) at  (1.41, 1.41){};\node (6) at  (2, 0){};\node (7) at  (1.41, -1.41){};
\node (8) at  (0,-2){};\node (9) at  (-1.41, -1.41){};
\draw[->] (2) to (1);\draw[->] (3) to (1);\draw[->] (4) to (1);
\draw[->] (5) to (1);\draw[->] (6) to (1);\draw[->] (7) to (1);
\draw[->] (8) to (1);\draw[->] (9) to (1);
\end{tikzpicture}
\end{subfigure}
\begin{subfigure}[b]{0.14\textwidth}
\centering
\begin{tikzpicture}
[scale = 0.5, every node/.style={circle, fill=white, draw=black}, gray node/.style = {circle, fill = blue, draw}]
\node (1) at  (0.59, 0.59){};\node (3) at  (0.59, 3.41){};
\node (5) at  (3.41, 3.41){};\node (7) at  (3.41, 0.59){};
\node (2) at  (0, 2){};\node (4) at (2, 4){};
\node (6) at (4, 2){};\node (8) at (2, 0){};
\draw[->] (1) [out=135,in=270] to (2);\draw[->] (2) [out=90,in=225] to (3);
\draw[->] (3) [out=45,in=180] to (4);\draw[->] (4) [out=0,in=135] to (5);
\draw[->] (5) [out=315,in=90] to (6);\draw[->] (6) [out=270,in=45] to (7);
\draw[->] (7) [out=225,in=0] to (8);\draw[->] (8) [out=180,in=315] to (1);
\end{tikzpicture}
\end{subfigure}
\begin{subfigure}[b]{0.14\textwidth}
\centering
\begin{tikzpicture}
[scale = 0.55, every node/.style={circle, fill=white, draw=black}, gray node/.style = {circle, fill = blue, draw}]
\node (1) at  (0.59, 0.59){};\node (3) at  (0.59, 3.41){};
\node (5) at  (3.41, 3.41){};\node (7) at  (3.41, 0.59){};
\node (2) at  (0, 2){};\node (4) at (2, 4){};
\node (6) at (4, 2){};\node (8) at (2, 0){};
\node (11) at (2, 1){};\node (10) at (2, 2){};
\node (9) at (2, 3){};
\draw[->] (1) [out=135,in=270] to (2);\draw[->] (3) [out=225,in=90] to (2);
\draw[->] (4) [out=180,in=45] to (3);\draw[->] (4) [out=0,in=135] to (5);
\draw[->] (5) [out=315,in=90] to (6);\draw[->] (7) [out=45,in=270] to (6);
\draw[->] (8) [out=0,in=225] to (7);\draw[->] (8) [out=180,in=315] to (1);
\draw[->] (8) to (11);\draw[->] (11) to (10);\draw[->] (9) to (10);\draw[->] (4) to (9);
\end{tikzpicture}
\end{subfigure}
\caption{From left to right: hub-spoke, cycle and linear-paths network. A directed arrow
determines the vertex that purchases the edge. We omit the details
of the regime of parameters ($\c$ and $p$) in which such networks can 
form in the equilibria of our game.
\label{fig:eq-examples}}
\end{figure}
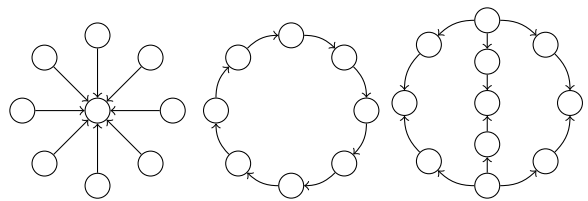

We remark that pure strategy equilibria exist in all parameter regimes of our game. 
When $\c \geq 1$, the empty network can form in equilibria for all $p$. When $\c<1$ 
a cycle or two disconnected hub-spoke structure of size $n/2$ can form in equilibria
depending on whether $p$ is far or close to 1 ($(1-\omega(1/n))$ and $(1-o(1))$, 
respectively).

Examples in Figure~\ref{fig:eq-examples} show that denser networks can form in equilibria 
compared to the model of~\citet{BalaG00} and the high immunization cost regime of the 
model of~\citet{GoyalJKKM16}. So it is natural to ask how dense equilibrium networks 
can be. We study this question in Section~\ref{sec:density} and show an upper bound of
$O(n \log n)$ on the density of the equilibrium networks. Since the examples in Figure~\ref{fig:eq-examples}
have $\Theta(n)$ edges, our upper bound is tight up to a logarithmic factor. 

Moreover, while all the 
equilibrium networks in Figure~\ref{fig:eq-examples} are connected, there might still exist
equilibrium networks in our game that are highly disconnected. In Section~\ref{sec:welfare}
we show that any equilibrium network with at least one edge contains a large connected component.
However, even with the guarantee of a large connected component, there might still be
concerns that the equilibrium networks can become highly fragmented after the attack. In Section~\ref{sec:welfare}
we show that as long as the equilibrium network is not too dense, the social welfare is lower bounded 
by $\Theta(n^2)$ i.e. a constant fraction of the social welfare achieved in the attack-free game.

We obtain these structural results even tough we cannot
compute utilities nor even verify that an equilibrium has reached
due to computational barriers. We view these results as are our
most significant technical contributions. 
\section{Edge Density}
\label{sec:density}
We now analyze the edge density of equilibrium networks. 
\begin{thm}
\label{thm:density1}
Any equilibrium network on $n$ vertices has 
$O\left(n\log n/p\right)$ edges.
\end{thm}

The proof of Theorem~\ref{thm:density1} is due 
to the following observations which 
we formally state and prove next. At a high level, 
we first show that if $G$ has large 
enough edge density, then $G$ contains an induced 
subgraph $H$ whose minimum cut size is large. We 
then show a large minimum cut size implies that $H[p]$ 
is connected with high probability. This 
means that in almost all attacks that infect a vertex in $H$, 
all vertices in $H$ will get infected.
So a vertex in $H$ would have a beneficial deviation in 
the form of dropping an edge; which 
contradicts the assumption that $G$ was an equilibrium 
network. This proves that equilibrium networks cannot be 
too dense.

More formally,
we first show in Lemma~\ref{lem:cut} that if $G$ is \emph{dense enough} 
it contains a subgraph $H$ with a 
minimum cut size, denoted by $\alpha(H)$, of at least 
$\Omega\left(\log n/p\right)$.
\begin{lemma}
\label{lem:cut}
Let $G=(V,E)$ be a graph on $n$ vertices. There exists a 
constant $k$ such that if  $|E| \geq k n \log n/p$ 
then $G$ contains an induced subgraph $H$ 
with $\alpha(H)\geq k\log n/p$.
\end{lemma}
\begin{proof}
If $\alpha(G)\geq k\log n/p$  then $G$ is the desired 
graph. Otherwise, there is a cut of size
less than $k \log n/p$ that partitions $G$ into two graphs 
$G_1$ and $G_2$. Repeat this process at
$G_1$ and $G_2$, and build a decomposition tree $T$ in 
this manner. Any leaf of this tree $T$ is
either a singleton vertex or a graph where the minimum 
cut size is at least $k \log n/p$. If
at least one leaf in $T$ satisfies 
the latter property, then we are done and this is our 
desired graph $H$. We now argue that it can not be the case 
that all leaf vertices in $T$ are singletons.

To see this, note that there can be at most $n-1$ 
internal vertices in $T$ and each internal
vertex in $T$ corresponds to removing up to 
$k \log n/p$ edges from $G$. Thus the 
decomposition process removes at most 
$k(n-1)\log n/p$ edges. On the other hand, 
$G$ has at least $k n \log n/p$ edges. 
It follows that not all leaves of $T$ can
be singleton vertices.
\end{proof}

We then show that if $\alpha(G)$ is 
$\Omega\left(\log n/p\right)$ then 
with high probability $G[p]$ is connected.
\begin{lemma}[\citet{noga}]
\label{lem:noga}
Let $G=(V,E)$ be a graph on $n$ vertices. 
Then for any constant $b>0$ there exists 
a constant $k(b)$ such that
if $\alpha(G) \geq k(b) \log n/p$ then with 
probability at least $1-n^{-b}$,
$G[p]$ is connected.\footnote{The statement in \citet{noga} requires
$\alpha(G[p]) \geq k(b) \log n$. Since $\alpha(G[p])=\alpha(G)p$
this translates to the condition stated in Lemma~\ref{lem:noga}.}
\end{lemma}

We now define a property which we call 
\emph{almost certain infection} and show
that no equilibrium network can contain an 
induced subgraph satisfying this property.
\begin{definition}
\label{def:certain}
Let $G=(V,E)$ be a graph on $n$ vertices and
let $H$ be a subgraph of $G$ on more than one vertex. 
$H$ has the \emph{almost certain infection} property if whenever 
any vertex in $H$ is attacked, then with probability at least
$1-o(1/n)$ the attack spreads to every vertex in $H$.
\end{definition}

\begin{lemma}
\label{lem:certain}
Let $G=(V,E)$ be an equilibrium network on $n$ vertices. 
$G$ cannot contain an induced subgraph $H=(V', E')$ such that 
$H$ satisfies the almost infection property.
\end{lemma}
\begin{proof}
Consider any equilibrium graph $G$ that 
violates the assertion of the claim. Let $H$
be an induced subgraph of $G$ with the 
almost certain infection property. 
We first prove that $H$ contains a cycle. 
Assume by the way of contradiction that $H$ does not 
have a cycle, so $H$ is a collection of trees.
Let $u$ be any leaf in $H$. Then $u$ is 
incident to at most one edge in $H$.
Therefore, with probability $1-p$, this edge is not 
in $G[p]$ and $u$ is not connected to any other 
vertices in $G[p]$. So $H$ cannot have the almost 
certain infection property
(Recall that we assumed $p$ is a constant 
independent of the number of players $n$). 
This means that $H$ contains a cycle $C$. Let $(u, v)$ 
be an edge on the cycle $C$. Assume
without loss of generality that $u$ purchased the edge $(u,v)$.

Now let $\xi$ be the event that an attack 
propagates to some vertex in $H$ after the attack.
Then conditioned on $\xi$, with probability at least
$1-o(1/n)$, all vertices in $H$ 
die. Hence vertex $u$ in $H$ has negative utility. On 
the other hand, if $\xi$ does not occur, then the utility of $u$ 
remains unchanged even if we remove the edge $(u, v)$. 
Thus vertex $u$ can strictly improve her utility in this case 
by dropping the edge $(u, v)$. 
A contradiction to the fact that $G$ is an equilibrium network.
\end{proof}

We are now ready to prove Theorem~\ref{thm:density1}.

\noindent\textit{Proof of Theorem~\ref{thm:density1}.}
Assume by way of contradiction that $G$ has more 
than $kn\log n/p$ edges where $k$ is the constant 
in Lemma~\ref{lem:cut}. Then 
by Lemma~\ref{lem:cut}, $G$ contains a subgraph 
$H=(V', E')$ such that $\alpha(H)\geq k\log n/p$. 
Since $|E'|\geq |V'|$, by Lemma~\ref{lem:noga}, 
$H$ has the almost certain infection property. 
However, $G$ cannot be an equilibrium network 
by Lemma~\ref{lem:certain}.
\qed

The most interesting regime for the probability of spread $p$
is when $p$ is a constant independent of $n$. While the upper
bound in Theorem~\ref{thm:density1} holds for all $p$, it becomes 
vacuous as $p$ gets small i.e. it becomes bigger than the 
trivial bound of $n^2/2$ when $p\leq k\log n/n$ 
for constant $k$. In Section~\ref{sec:edge-small-p}
we analyze the edge density of equilibrium networks 
in the regime where $p < 1/n$. We 
show that the number of edges in any equilibrium 
network is bounded by $O(n)$ in this regime. 
To prove the density result we utilize properties of 
the Galton-Watson branching process and 
random graph model of Erd{\"{o}}s-R\'enyi, as well 
as tools from extremal graph theory.
\subsection{Small $p$ Regime}
\label{sec:edge-small-p}
In this section we focus on the regime where $p < 1/n$
and prove the following upper bound on the edge density. 

\begin{thm} 
    \label{thm:smallp}
    Let $p=\kappa/n$ for some constant $\kappa <1$.
    Let $G=(V,E)$ be an equilibrium network over $n$ vertices.
    Then for sufficiently large $n$, $|E|\le \max\{1/\c, 24000\} n$.

\end{thm}

We prove Theorem~\ref{thm:smallp} by contradiction and show that
if the equilibrium graph has more than $\max\{1/\c, 24000\} n$ edges, 
there exists a beneficial deviation in the form of dropping an edge
for one of the players.

In order to prove Theorem~\ref{thm:smallp}, we need structural results
stated in Lemmas~\ref{lem:cut2}~and~\ref{lem:logc}.
First, consider an edge $(u, v)$ purchased by vertex $u$. Purchasing this edge
would not have increased the connectivity benefit of $u$ unless, after
some attack, the edge $(u,v)$ is the only path connecting $u$ to 
$v$ (and possibly other vertices that are only reachable through $v$). 
In Lemma~\ref{lem:cut2} (which we will prove later) we show that if a
graph is dense enough, then there exists an edge $(u, v)$ such that 
many vertices should be deleted in order to make $(u,v)$ the 
only remaining path connecting $u$ and $v$.

\begin{lemma}
\label{lem:cut2}
Let $G=(V,E)$ be a graph on $n>3 \gamma$ vertices with 
$|E| \geq 2.5\gamma(n-\gamma)-1$ for some $\gamma \in \N$. Then 
there exist two vertices $v_1$ and $v_2 \in V$ such that $(v_1, v_2)\in E$,
and at least $\gamma+1$ vertices need to be deleted so that the only
path from $v_1$ to $v_2$ in $G$ is through the direct edge $(v_1, v_2)$.
\end{lemma}

Second, as described in Section~\ref{sec:model}, the number of vertices that are
killed in any attack is the size of the connected component in $G[p]$ that 
contains the initially attacked vertex. 
Lemma~\ref{lem:logc} (which we will prove later) bounds the size of a randomly 
chosen connected component in $G[p]$.
\begin{lemma} \label{lem:logc} 
Let $G=(V,E)$ be an equilibrium network over $n$ vertices
with $|E|=kn$ and $\max\{1/\c, 24000\}\le k = O(\log n)$. 
When $p<1/n$ and $n$ is sufficiently large,
the size of the connected component of a randomly chosen 
vertex $v$ in $G[p]$ is at most 
$k/3$ with probability at least $1-2\c/(3n)$.
\end{lemma}

We now give the formal proof of Theorem~\ref{thm:smallp}.
\begin{proof}[Proof of Theorem~\ref{thm:smallp}]
Assume by way of contradiction that $G$ is an equilibrium network with 
$kn$ edges where $k> \max\{1/\c,24000\}$. 
Let $\gamma = \lfloor k/2.5 \rfloor$, we have $kn>2.5\gamma(n-\gamma)-1$.
By Lemma~\ref{lem:cut2}, there exists an edge $(u,v)$ such that in order to
make $(u,v)$ the only remaining path connecting $u$ and $v$, we need to delete 
at least $\gamma+1$ vertices. Without loss of generality assume that 
$u$ has purchased the edge $(u,v)$. If in any attack at most $\gamma$
vertices are killed, $u$ will not lose any connectivity benefit 
after dropping this edge but decrease her expenditure by $\c$.

Consider the size of the largest connected component in $G[p]$. Since $\kappa<1$,
the size of the largest connected component in $G[p]$ is at most $\beta \log n$ with probability $1-o(1/n)$
for sufficiently large constant $\beta$ which only depends on $\kappa$.  When 
$G$ is the complete graph, 
$G[p]$ corresponds to the random graph generated by the 
Erd{\"{o}}s-R\'enyi model. In such case, the size of the largest
component of $G[p]$ is $O(\log n)$, with high probability, when 
$p n = \kappa <1$ \cite{erdos1960evolution}.

If $\gamma>\beta \log n$, then with probability at most $o(1/n)$,
the attack kills more than $\gamma$ vertices, in which case the connectivity 
benefit of $u$ can decrease by at most $n$ after
dropping the edge $(u,v)$.
So the expected connectivity benefit of $u$ decreases by at most $o(1)$ 
after the deviation but her expenditure also decreases by $\c$. 
Hence, after the deviation, the expected change in the utility of $u$ is at least $\c - o(1)>0$ 
which contradicts the assumption that $G$ is an equilibrium network
(recall that we have assumed $\c$ is a constant independent of $n$).

If $\gamma \le \beta \log n$, then by definition, 
$k = O(\log n)$. Since $k$ is also at least $\max \{ 24000,1/\c \}$,
by Lemma~\ref{lem:logc}, with probability at least $1-2\c/(3n)$,
the attack kills at most $k/3<\gamma$ vertices, in which case the connectivity
benefit of $u$ remains unchanged after dropping the edge $(u,v)$.
With probability at most $2\c/(3n)$, more than $k/3$ vertices are killed
in which case the connectivity benefit of $u$ can decrease by at most $n$.
So the expected connectivity benefit of $u$ decreases by at most $2\c/3$ after
the deviation but her expenditure also decreases by $\c$. 
Hence, after the deviation, the expected change in the utility of $u$ is at least $\c/3>0$, 
which contradicts the assumption that $G$ is an equilibrium network.
\end{proof}

We now proceed to prove Lemmas~\ref{lem:cut2}~and~\ref{lem:logc}.

\begin{proof} [Proof of Lemma~\ref{lem:cut2}]
    Assume by way of contradiction that $G$ is the graph with 
    smallest number of vertices such that
    $G$ has $n>3\gamma$ vertices and at least $2.5\gamma(n-\gamma)-1$
    edges. 
    Therefore, $G$ has two vertices $v_1$ and $v_2\in V$ such that $(v_1, v_2)\in E$
    and there exists a vertex set $S\subset V$ with at most $\gamma$ vertices, such that 
    after deleting $S$ the only
    path from $v_1$ to $v_2$ in $G$ is through the edge $(v_1, v_2)$.
    If $S$ has less than $\gamma$ vertices, then we add arbitrary vertices from $V$ 
    (but not $v_1$ or $v_2$) to $S$ so that $S$ has exactly $\gamma$ vertices 
    (we can always do so because $G$ has more than $3\gamma$ vertices).

    Consider the graph where the edge $(v_1, v_2)$ and the vertices in $S$ are removed,
    $v_1$ and $v_2$ are not connected in this graph. Let $C_1$ be the connected component
    that contains $v_1$ and $C_2 = V \setminus S \setminus C_1$. By definition,
    $(v_1,v_2)$ is the only edge between $C_1$ and $C_2$ in $G$.

    Define two graphs $G_1 = (V_1, E_1)$ and $G_2 = (V_2, E_2)$ as subgraphs of $G$
    induced by $C_1 \cup S$ 
    and $C_2 \cup S$. Suppose $G_1$ has $n_1$ 
    vertices and $G_2$ has $n_2$ vertices where $\gamma+1\leq n_1\leq n-1$ and $\gamma+1\leq n_2\leq n-1$. 
    Also without loss of generality assume $n_1\geq n_2$. 
    We have that $n_1+n_2 = n+ \gamma$. 
    On the other hand, $G_1$ and $G_2$ have at least $2.5\gamma(n_1+n_2-2\gamma)-2$ 
    edges in total (any edge which is not ($v_1$,$v_2$) is either in $G_1$ or $G_2$). 
    So either $G_1$ has at least $2.5\gamma(n_1-\gamma)-1$ edges or $G_2$ has at least $2.5\gamma(n_2-\gamma)-1$ edges. 
    Also by the property of $G$, 
    for any pair of vertices $v_1', v'_2\in V_1$ (or $v_1', v'_2\in V_2$) such that $(v'_1, v'_2)\in E_1$
    (or $(v'_1, v'_2)\in E_2$) we only need to delete 
    at most $\gamma$ vertices so that the only
    path from $v'_1$ to $v'_2$ is through the direct edge $(v'_1, v'_2)$.

    We claim that there exists a graph $G'$ (which is either $G_1$ or $G_2$) with $n'>2\gamma$ 
    vertices that has at least $2.5\gamma(n'-\gamma)-1$ edges. 
    Note that if $G_1$ has at least $2.5\gamma(n_1-\gamma)-1$ edges then we are done
    since $n_1 \geq n_2$ and $n_1+n_2 > 4 \gamma$ imply that $n_1 > 2\gamma$.
    So suppose 
    $G_1$ has less than $2.5\gamma(n_1-\gamma)-1$ edges. Therefore,
    $G_2$ has at least $2.5\gamma(n_2-\gamma)-1$ edges. Again if $n_2 > 2\gamma$
    we are done so suppose $n_2\leq 2\gamma$.
    Consider the edges which are in $G_2$ but not in $G_1$. 
    These edge have at least one endpoint in $C_2$, so there are at most 
    \begin{align*}
        \frac{(n_2-\gamma)(n_2-\gamma-1)}{2}+\gamma(n_2-\gamma)  &< (n_2-\gamma)(\frac{n_2}{2}+\frac{\gamma}{2})\\
                                                     &\leq \frac{3\gamma}{2}(n_2-\gamma) \\
                                                     & \leq 2.5\gamma(n_2-\gamma)  - 1
    \end{align*}
    such edges. But this would imply 
    $G_1$ has at least $2.5\gamma(n_1-\gamma)-1$. 

    If $G'$ has strictly more than $3\gamma$ vertices, it contradicts our assumption that $G$ is 
    the smallest graph with the property stated in the lemma. If $G'$ has at most $3\gamma$ vertices, 
    then $G'$ has 
    at most 
    \begin{align*}
        \frac{n'(n'-1)}{2} &= \frac{(n'-2\gamma)(n'+2\gamma)}{2} +2\gamma^2 - \frac{n'}{2} \\
                           &\le2.5 \gamma(n'-2\gamma) + 2 \gamma^2  - \frac{n'}{2}\\
                           & < 2.5 \gamma (n'-\gamma)-1
    \end{align*}
    edges, which is a contradiction.  
\end{proof}

Before proving Lemma~\ref{lem:logc},    
let us introduce some notation. Let $H\subseteq V$ 
be the set of vertices in $G$ with degree at least $n^{3/4}$. Also let $L=V\setminus H$ be the 
set of  vertices in $G$ with degree strictly less than $n^{3/4}$. Hence, $H$ and $L$ correspond
to vertices with \emph{high} and \emph{low} degrees in $G$, respectively.
Recall that $G[p]$ is a random graph where each edge of $G$ is sampled independently  
to be retained in $G[p]$ with probability $p$.
Using $H$ and $L$, the creation of $G[p]$ can be describe as a three step sampling process. 
In the first step, edges with both endpoints in $H$ are sampled to be retained. 
In the second step, edges with both endpoints in $L$ are sampled to be retained.
Finally, in the third step, edges with one endpoint in $H$ and the other endpoint in $L$
are sampled to be retained.

In Lemma~\ref{lem:higd}, we first show 
that with high probability the size of the largest connected component created 
by the first step and second step of the sampling process is at most 5 and 12, respectively. 
These connected components can then be connected together in the third step of the sampling
to create larger connected components in $G[p]$.
We show that with high probability, the third step would not connect more than 
3 of the connected components of high degree vertices (which were created in the first step).
This implies that the number 
of high degree vertices in any connected component of $G[p]$ is at most $15$ with 
high probability.

We then show in Lemma~\ref{lem:degdist} that for any $\alpha>0.001$, the expected 
number of vertices with at least $\alpha k$ edges 
in $G[p]$ is at most $\c/(2^{1200\alpha}k)$. We then use the structural results of
Lemma~\ref{lem:higd} and~\ref{lem:degdist} to show that 
with probability at least $1-2\c/(3n)$, the size of the connected component of a 
randomly chosen vertex in $G[p]$ is at most $k/3$ .

\begin{lemma} \label{lem:higd}
    Let $G=(V,E)$ be an equilibrium network over $n$ vertices
    with $|E|=kn$ with $k \ge 24000$ and $k=O(\log n)$.
    Suppose $p<1/n$ and $n$
    is sufficiently large. Then, with probability at least $1-o(1/n)$,  (1) the connected
    components generated in the first step and the second step of the sampling process
    of creating $G[p]$ have size
    at most $5$ and $12$, respectively and (2)
    no component in $G[p]$ has more than 15 vertices from the set $H$.
\end{lemma}
\begin{proof}
    Recall that in the first two 
    steps of creating the $G[p]$, we sample 
    edges to retain  independently with probability 
    $p$ from the graphs induced 
    by only $H$ and $L$, respectively. Since 
    $|E|=kn$ the number of high 
    degree vertices is at most $\tilde{O}(n^{1/4})$ 
    (where the notation $\tilde{O}$
    hides logarithmic dependencies on $n$). 
    So any vertex in the graph induced 
    by $H$ has degree at most 
    $\tilde{O}(n^{1/4})$. Moreover, by definition, the 
    vertices in the graph induced by $L$ 
    have degree at most $n^{3/4}$. By 
    Corollary~\ref{cor:gw} (in Appendix~\ref{sec:useful-lem}), with high probability, the random 
    components formed by the first step and 
    second step of the sampling process have size at most 5 and 12, respectively.
    We refer to the set of the components of $G[p]$ that are 
    formed by step one and two of the sampling 
    by $\mathcal{C}_1$ and $\mathcal{C}_2$, respectively.

    Consider a component $C_2\in \mathcal{C}_2$. The probability 
    that there is an edge in $G$ between any vertex in $C_2$
    and a specific high degree vertex is bounded by $12/n$. 
    This means that the probability that the 
    vertices in $C_2$ are connected to more than one  
    high degree vertex is at most 
    $\tbinom{|H|}{2}(12/n)^2 = \tilde{O}(n^{-1.5})$.
    Similarly the probability that the 
    vertices in $C_2$ are connected to more than two
    high degree vertices is at most 
    $\tbinom{|H|}{3} (12/n)^3 = o(n^{-2})$. 
    Therefore, with high probability, there is no 
    component $C_2\in \mathcal{C}_2$ that is connected to three high degree vertices.
    Moreover, the probability that there are three connected components
    in $\mathcal{C}_2$ that is connected to $2$ high degree vertices
    is at most $\tbinom{n}{3} (\tilde{O}(n^{-1.5}))^3 = o(n^{-1})$.

    In the third step of creating the $G[p]$, components 
    from $\mathcal{C}_1$ and $\mathcal{C}_2$ would become connected
    by sampling the edges in between $H$ and $L$. 
    As we showed there are
    most two components in $\mathcal{C}_1$ that will be 
    connected in the $G[p]$ by the edges sampled
    in the third step.
    This means, with high probability, no component in 
    $G[p]$ will include more than 
    3 components from $\mathcal{C}_1$; so, with high probability, 
    no component in $G[p]$ has more than 15 high degree vertices. 
\end{proof}

\begin{lemma} \label{lem:degdist}
    Let $G=(V,E)$ be an equilibrium network over $n$ vertices
    with $|E|=kn$ and $k\ge \max\{24000,1/\c\}$.
    For any $\alpha>0.001$, when $p<1/n$, the 
    expected number of vertices 
    that have at least $\alpha k$ adjacent edges in 
    $G[p]$ is at most $\c/(2^{1200\alpha}k)$
    for sufficiently large $n$. 
\end{lemma}

\begin{proof}
    For a vertex with degree $d>\alpha k$, the probability that
    she has $\alpha k$ edges in $G[p]$ is 
    at most
    \begin{align*} 
        \tbinom{d}{\alpha k} n^{-\alpha k} 
        &< (\frac{de}{n\alpha k})^{\alpha k} 
        < \frac{d}{n} (\frac{e}{\alpha k})^{\alpha k}\\
        &< \frac{d}{n} \cdot 8^{-\alpha k}
        \le \frac{d}{n}  \cdot 2^{-2\alpha k} \cdot 2^{-24000\alpha}\\
        &< \frac{d}{n} \cdot  k^{-3} \cdot 2^{-24000\alpha}
        < \frac{d}{n} \cdot \frac{\c}{2k^2} 2^{-1200\alpha}.
    \end{align*}
    The first inequality is due to Stirling's formula. 
    Other inequalities are due to 
    $k\alpha\ge 24$, $k\ge \max\{24000,1/\c\}$ and $e/(\alpha k) < 1/8$.
    Adding up the probabilities for all the vertices 
    (using linearity of expectation) 
    and using the fact that the sum of 
    the degrees of all the vertices is $2kn$ will 
    conclude the proof.
\end{proof}

We now have all the background to prove 
Lemma~\ref{lem:logc}.

\noindent\textit{Proof of Lemma \ref{lem:logc}.}
If we randomly choose a vertex $v$, the probability 
that the connected 
component of $v$ in $G[p]$ has size at least $k/3$
is upper bounded by $1/n$ times the sum of the sizes 
of components 
with size at least $k/3$ in $G[p]$. 
This is in turn upper bounded by 
$$\frac{1}{n}\sum_{i=1}^n (i+1) \frac{k}{3} (x_i-x_{i+1}) 
< \frac{1}{n}\sum_{i=1}^n k x_i,$$ 
where $x_i$ is the number of components with size at 
least $ik/3$.

Recall that we partitioned  the vertices of $G$ into high 
and low degree vertex sets $H$
and $L$ based on the degree. We described a 
three step sampling process for creating $G[p]$
and referred to the set of connected components of $G[p]$ that are 
formed by step one and two of the sampling 
by $\mathcal{C}_1$ and  $\mathcal{C}_2$, respectively.
Let $\xi$ be the event such that each $C_2\in \mathcal{C}_2$ 
has at most 12 vertices 
and each component in $G[p]$ has at most 
$15$ high degree vertices. 
By Lemma~\ref{lem:higd}, $\Pr[\bar{\xi}] = 
1-\Pr[\xi]=o(1/n)$. Let $y_i$ be the number of vertices 
with at least $ik/600$ edges in $G[p]$. By 
Lemma~\ref{lem:degdist}, $\E[y_i]\leq \c/(2^{2i} k) \le \c/(2^{i+1}k)$.

Fix a component of $G[p]$.
Conditioned on event $\xi$, if all the high degree vertices in the component have 
degree at most $i k/600$, then each high degree vertex
is connected to at most $ik/600$ of the components in $\mathcal{C}_2$. 
So the size of this component is at most $15 \cdot 12 \cdot (ik/600) + 15 = 3ik/10 + 15 < (ik/3)$. 
Therefore, each component with size at least $ik/3$ contains at least one 
vertex with at least  $i k/600$ adjacent edges in $G[p]$. 
This means that $\E[x_i|\xi] \le \E[y_i|\xi]$. So
\begin{align*}
    \E[x_i|\xi] &\le \E[y_i|\xi] 
    = \frac{\E[y_i]- \E[y_i|\bar{\xi}] \Pr[\bar{\xi}]}{\Pr[\xi]}\\
    &\le \frac{\E[y_i]}{\Pr[\xi]} 
    \le \frac{\c}{(2^{i+1} k)(1-o(1/n))} \\
    &= \frac{\c}{2^{i+1}k} \cdot (1+o(1/n)).
\end{align*}

So the probability that $v$ is in a component of size at least $k/3$ given $\xi$ is 
\begin{align*}
    \frac{1}{n}\E[\sum_{i=1}^n k x_i|\xi] 
    &\le \frac{1}{n}\sum_{i=1}^n k \E[x_i| \xi] \\
    &= (1+o(1/n))\frac{1}{n}\sum_{i=1}^n \frac{\c}{2^{i+1}}\\
    &= (1+o(1/n))\frac{\c}{2n}.
\end{align*}

So the overall probability that a vertex is in a component with 
size at least $k/3$ is at most $(1+o(1/n))\c/(2n)+\Pr[\bar{\xi}]<2\c/(3n)$.
\qed

\section{Social Welfare}
\label{sec:welfare}

In this section we provide a lower bound on the social welfare of equilibrium networks.
Similar to other reachability games, the empty graph can form in 
equilibrium~\cite{BalaG00, GoyalJKKM16}. 
Hence without any further assumptions, 
no meaningful guarantee on the social welfare can be made. Hence, we focus on non-trivial 
equilibrium networks defined as follows.
\begin{definition}
\label{def:non-trivial}
An equilibrium network is non-trivial if it contains at least one edge.
\end{definition}

Definition~\ref{def:non-trivial} rules out the empty network but it is still possible that
a non-trivial equilibrium network contains many small connected components or 
becomes highly fragmented after the attack. In this section we show that none
of these concerns materialize. In particular, in 
Theorem~\ref{lem:largest}, we first show  that 
any non-trivial equilibrium network contains at least one large connected
component. We then show in Lemma~\ref{lem:small-components} 
that when the network is not too dense,
the equilibrium network cannot become highly fragmented after the attack.
These two observations allows us to prove our social welfare
lower bound as stated in Theorem~\ref{thm:welfare}.

We start by showing that any non-trivial equilibrium network contains a 
large connected component. 
\begin{thm} \label{lem:largest}
Let $G=(V,E)$ be a non-trivial equilibrium network over $n$ vertices.
Then, for sufficiently large $n$, the largest connected component of $G$ has at least 
$n/3$ vertices.
\end{thm}
\begin{proof}
Throughout we require $n > \max\{15\c+15\c^2, (9+3\c)/2, \c^2+4\c+5\}$.
Consider two cases: (1) when $G$ contains no isolated vertices\footnote{An
isolated vertex is a vertex with no incident edges.}
and (2) when $G$ contains at least one isolated vertex.

In the first case, assume by way of contradiction that $G$ has at least $4$ connected components.
Let $C_0$ be a smallest connected component in $G$, say of size $n_0$.
Also let $C_1$ be a second-smallest connected component in $G$, say of size $n_1$.
By construction, $n_0\le n_1$. 
Since $C_1$ has at least one edge and this edge has cost $\c$,
the size of $C_1$ is at least $\c +1$. 
Otherwise the vertex who bought this edge can improve her utility by dropping this edge.

Consider the deviation that a vertex $v$ in $C_0$ adds an edge to an arbitrary vertex in $C_1$. 
We show that the increase in the connectivity benefit of $v$ is more than $\c$. 
If the attack does not start at $C_0$ or $C_1$, which occurs with probability $(n-n_0-n_1)/n$, 
then the connectivity benefit of $v$ increases by $n_1$.
If the attack starts at $C_0$, then the only way for the  attack to reach $C_1$ is through 
the newly added edge by $v$. Hence, in this case the connectivity benefit of $v$ does not decrease. 
Therefore, the only scenario in which the connectivity benefit of $v$ can decrease is when 
the attack starts at $C_1$
(which occurs with probability $n_1/n$). In this case, the connectivity benefit of $v$ can decrease 
by at most $n_0$. So the change in the connectivity benefit of $v$ is at least  
$$\Delta \geq \frac{(n-n_0-n_1)}{n}n_1-\frac{n_1}{n}n_0 = \frac{n_1(n-2n_0-n_1)}{n},$$
after the deviation. We show that $\Delta > \c$ which is a contradiction to $G$ being an equilibrium network.

We consider two sub-cases based on the value of $n_1$: (1a) $n_1>5\c$ and (1b) $n_1\le 5\c$.
First consider case (1a) where $n_1>5\c$.
Since $G$ has at least $4$ connected components, then $n_0\le n/4$, 
$n_0+n_1\le n/2$ and $2n_0+n_1\le 3n/4$.
So $$n_1(n-2n_0-n_1)>5\c \frac{n}{4} > \c n \implies \Delta > \c.$$ 
Next consider case (1b) where $n_1\le 5\c$. Since $n-2n_0-n_1 \ge n-15\c$ 
and $n_1\geq \c+1$ then
\begin{align*}
n_1(n-2n_0-n_1) &\ge (\c+1)(n-15\c) \\&= \c n + n -15\c - 15\c^2 > \c n\\
&\implies \Delta > \c,
\end{align*}
when  $n>15\c+15\c^2$.

Therefore, the deviation of adding an edge by $v$ will increase the connectivity benefit
of $v$ by strictly more than $\c$ which is a contradiction. So $G$ contains at most $3$ 
connected components in case (1) and hence
the largest connected component of $G$ contains at least $n/3$ vertices.

In case (2), we show that the largest connected component of $G$ contains at 
least $n-\c-3$ vertices. Therefore, the largest connected component 
of $G$ contains at least $n/3$ vertices when $n > (9+3\c)/2$.

Let $v$ be an isolated vertex and let $C^{\star}$ be a largest connected component of $G$, say of size
$n^{\star}$. Consider the deviation where $v$ adds an edge to an arbitrary vertex in $C^{\star}$.
If the attack does not start neither in $C^{\star}$ nor at $v$ (which occurs with probability $(n-n^{\star}-1)/n)$,
then the connectivity benefit of $v$ increases by $n^{\star}$.
The only scenario in which the connectivity benefit of $v$ can decrease is when the attack starts at $C^{\star}$
(which occurs with probability $n^{\star}/n$).
In this case, the connectivity benefit of $v$ can decrease by at most $1$. So the change $\Delta$ in the 
connectivity benefit of $v$ after the deviation is at least 
$$\Delta \geq \frac{(n-n^{\star}-1)}{n}n^*-\frac{n^{\star}}{n} = \frac{n^*(n-n^*-2)}{n}.$$

Note that $\Delta \leq \c$, since $G$ is an equilibrium network.
This implies that $n^*(n-n^*-2) \le \c n$. To show that $n^* > n-\c-3$,
consider the function $f(x) = x(n-x-2)$. $f$ is increasing when $x\leq (n-2)/2$
and decreasing when $x\geq (n-2)/2$. Moreover, 
$$f(\c+1) = f(n-\c-3) = (\c+1)(n-\c+3) > \c n$$ 
when $n > \c^2+4\c+5$. So $f(n^*)$ is always larger than $\c n$
when $\c+1\leq n^*\leq n-3-\c$. Since we showed $f(n^*) = n^*(n-n^*-2) \le \c n$
in equilibrium it most be the case that either $n^* < \c+1$ or $n^* > n-3-\c$.
The former cannot happen because we assumed $G$ is non-empty, so $C^*$
must have at least one edge and therefore $n^* \geq \c+1$. Hence, $n^* > n-3-\c$
as claimed.

\end{proof}

We next present Lemma~\ref{lem:small-components} that describes
the relationship between the expected size of the largest connected component of $G[p]$ 
and the connectivity benefits of the vertices in $G$.
\begin{lemma}
\label{lem:small-components}
Let $G=(V,E)$ be an equilibrium network over $n$ vertices. Let 
$C$ be any connected component in $G$ of size $n_C$.
If the expected size of the largest component of $C[p]$ is at most $n_C(1-\epsilon)$, 
then the expected sum of the connectivity benefits of the vertices in $C$
is at least $(n-n_C)n_C^2/n+\epsilon n_C^3/(3n)$.
\end{lemma}
\begin{proof}
With probability $(n-n_C)/n$, the attack starts at a vertex outside of $C$. In this
case the sum of the connectivity benefits of vertices in $C$ is $n_C^2$ 
(the first term in the lower bound). 
Otherwise, with probability $n_C/n$, the attack starts at a vertex in $C$. 
We claim that in this case, the sum of the connectivity benefits of the vertices 
in $C$ is at least $\epsilon n_C^2/3$ (the second term in the lower bound).

To prove the claim it suffices to show that if the largest component in $C[p]$ has size $X$, 
then the  sum of connectivity benefits of the vertices in $C$ is at least $(n_C-X) n_C/3$. 
The claim would then follow by taking the expectation and using the assumption of the theorem 
that $X=n_C(1-\epsilon)$.

Assume by way of contradiction that the sum of connectivity benefits of vertices in 
$C$ is less than $(n_C-X)n_C/3$.
Then there exists a connected component $C_0$ in $C[p]$ such that if we delete the vertices in $C_0$,
then the sum of connectivity benefits of the vertices when the attack 
destroys $C_0$ is less than $(n_C-X) n_C/3$. 
Suppose $C_0$ has size $n_0$ and we know $n_0\le X$.
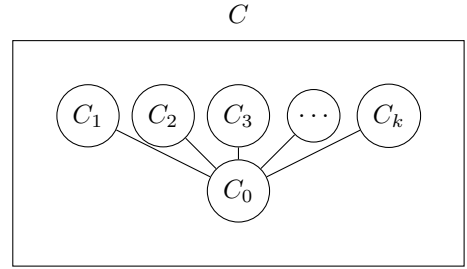
\begin{figure}[ht!]
\centering
\begin{tikzpicture}
[scale=0.50, every node/.style={circle,fill=white, draw=black}, gray node/.style = {circle, fill = blue, draw}]
\node [draw,rectangle,fill=white,minimum width=6cm,minimum height=3cm,label=$C$] at (0, 1) {};
\node (1) at  (0, 0){$C_0$};
\node (2) at  (-4, 2){$C_1$};\node (3) at  (-2, 2){$C_2$};\node (4) at  (0, 2){$C_3$};
\node (5) at  (2, 2){$\ldots$};\node (6) at  (4, 2){$C_k$};
\draw (1) to (2);\draw (1) to (3);\draw (1) to (4);\draw (1) to (5);\draw (1) to (6);
\end{tikzpicture}
\caption{\label{fig:proof}The connected component $C$ which contains components 
$C_0, \ldots, C_k$.}
\end{figure}

Let $C_1, C_2,\dots, C_k$ be the connected components in the subgraph 
of $G$ induced by $C \setminus C_0$ and 
let $n_1,n_2,\dots,n_k$ denote their sizes, respectively (see Figure~\ref{fig:proof}). 
Then by the assumption on the sum of connectivity benefits of the vertices after deleting $C_0$, 
when the attack destroys $C_0$ we have that
\begin{equation}
\label{eq:x2}
\sum_{i=1}^k n_i^2 < \frac{(n_C-X) n_C}{3}.
\end{equation}

If the attack starts at a vertex in a component $C_i$, then the vertices in $C \setminus C_i$ 
will still remain connected. 
This means the sum of connectivity benefits of the vertices in $C$ is at least 

\begin{align*}
\sum_{i=1}^k \frac{n_i}{n_C} (n_C-n_i)^2 &= \sum_{i=1}^k \frac{n_i}{n_C}(n_C^2-2n_Cn_i+n_i^2) \\
&\ge \sum_{i=1}^k n_Cn_i - 2\sum_{i=1}^k n_i^2\\
&=n_C\sum_{i=1}^k n_i - 2\sum_{i=1}^k n_i^2 .
\end{align*}

Since $\sum_{i=1}^k n_i=n_C-n_0 \ge n_C-X$ and $\sum_{i=1}^k n_i^2 < (n_C-X) n_C/3$ 
by Equation~(\ref{eq:x2}), the sum of the connectivity
benefits of the vertices in $C$ is at least $(n_C-X) n_C/3$; which is a contradiction.
\end{proof}

Theorem~\ref{lem:largest}~and~Lemma~\ref{lem:small-components} allow us 
to prove a lower bound on the social welfare of 
non-trivial equilibrium networks.

\begin{thm}
\label{thm:welfare}
Let $G=(V,E)$ be a non-trivial equilibrium network over $n$ vertices.
For any $\epsilon\in(0,1/8)$ and sufficiently large $n$, if $|E| < (1-2\epsilon)n\log(1/\epsilon)/p$, 
then the social welfare of $G$ is at least $\epsilon n^2/3 - O(n/p)$.
\end{thm}
\begin{proof}
We show the expected sum of the connectivity benefits of the vertices 
in $G$ is at least $\epsilon n^2/3$.
Subtracting off the cumulative expenditure for edge purchases which is $|E|\c = O(n/p)$
then imply the statement of the theorem. 

Suppose the largest connected component of $G$ say $C^*$ 
has size $n^{\star}$. By Theorem~\ref{lem:largest},
$n^{\star}\ge n/3$. We consider two cases based
on the size of $n^*$: (1) $n^{\star} \le (1-2\epsilon)n$ and (2) $n^* > (1-2\epsilon)n$.

In case (1), where $n^{\star} \leq (1-2\epsilon)n$, the sum of the connectivity 
benefits of the vertices in the largest connected component is at least $(n-n^{\star})n^{\star 2}/n$.
This is because with probability of $(n-n^{\star})/n$ the attack starts outside of the component 
and all the vertices in the component survive; in such case the sum of connectivity benefits
of the vertices in $C^*$ is $n^{\star 2}$. Moreover, the derivative of $(n-n^{\star})n^{\star 2}/n$
with respect to $n^{\star}$ is $(-3{n^{\star}}^2+2n n^{\star})/n$, which is positive when $0<n^{\star}<2n/3$
and negative when $n^{\star}>2n/3$. Since $n^{\star} \in [n/3,(1-2\epsilon)n]$, the minimum value of 
$(n-n^{\star})n^{\star 2}/n$ should be at one of the end points which correspond to values 
$2n^2/27$ or $2\epsilon(1-2\epsilon)^2n^2$, respectively. Both of these values
are larger than $\epsilon n^3/3$ when $\epsilon<1/8$, which means the sum of connectivity benefits 
is at least $\epsilon n^2/3$ in this case.

In case (2), where $n^* > (1-2\epsilon)n$, the number of edges in the connected component $C^*$ is 
at most $n^{\star}\log(1/\epsilon)/p$ (which occurs when all the edges are in this component).
Let us denote the vertices in $C^*$ by numbers from $1$ to $n^{\star}$. Let $d_i$ denote the 
degree of vertex $i$. We first bound the expected number of isolated 
vertices of $C^*$ in $G[p]$.
A vertex becomes isolated in $G[p]$ if none of the edges 
adjacent to it are sampled to  be retained. This event occurs with probability $(1-p)^{d_i}$
for vertex $i$. So we can derive a lower bound on the expected number of isolated vertices of $C^*$ in $G[p]$
as follows.
\begin{align*}
\sum_{i=1}^{n^{\star}} (1-p)^{d_i}\geq n^* (1-p)^{\frac{2|E|}{n^{\star}}}  > \epsilon n^{\star},
\end{align*}
where the first inequality is by inequality of arithmetic and geometric means and the second inequality is
by the assumption that $|E| < n^{\star}\log(1/\epsilon)/p$. Thus
the expected size of the largest connected component in $C^*[p]$ is at most $n^*(1-\epsilon)$.
We can now apply Lemma~\ref{lem:small-components}
to show that the expected sum of the connectivity benefits of the vertices in $C^*$ is at least 
$$\frac{\epsilon n^{\star 3}}{3n}+ \frac{(n-n^{\star})n^{\star 2}}{n} = \frac{n^{\star 2}n-(1-\epsilon/3)n^{\star 3}}{n}.$$
This is strictly decreasing in $n^*$ as the derivative with respect to $n^*$ is negative. So
the expected sum of connectivity benefits of the vertices in $C^*$ (and hence in $G$) is 
at least $\epsilon n^2/3$ (when $n^{\star}=n$).
\end{proof}

Finally, we remark that unlike the models of~\citet{BalaG00}~and~\citet{GoyalJKKM16},  
achieving a social welfare of $n^2-o(n^2)$ is impossible in our game even when restrciting
to sparse and non-trivial 
equilibrium networks. This is formalized in Proposition~\ref{lem:upper-welfare}. 
\begin{pro}
\label{lem:upper-welfare}
There exists a non-trivial equilibrium network $G=(V,E)$ over $n$ vertices with $O(n)$ edges 
such that the social welfare of $G$ is $kn^2$ for $k <1 $. 
\end{pro}
\begin{proof}
The hub-spoke equilibrium (see Figure~\ref{fig:eq-examples})
with $p = 0.6$ satisfies the condition of Proposition~\ref{lem:upper-welfare}
with $k=0.4$.
\end{proof}
\section{Conclusions}
\label{sec:future}
We studied a natural network formation game where each network connection 
has the potential to both bring additional utility to an agent as well add to her 
risk of being infected by a cascading infection attack. We showed that the 
equilibria resulting from these competing concerns are essentially sparse
and containing at most $O(n \log n)$ edges. We also  showed that any 
non-trivial equilibrium network in our game achieves the highest possible 
social welfare of $\Theta(n^2)$ whenever the equilibrium network has only 
$O(n)$ edges.

The Price of Anarchy in our model is $\Theta(n)$. 
To illustrate, consider the $\c \geq 1$ regime. A central planner can built 
a cycle or two disconnected hub-spoke structures of size $n/2$ depending on 
whether the probability of spread of the attack $p$ is low or high, respectively 
(and both of these structures can also form in equilibrium). Such networks 
have social welfare of $\Theta(n^2)$. However, the empty network is an 
equilibrium network in this regime implying a Price of Anarchy of at least 
$\Theta(n)$ -- the worst Price of Anarchy possible. The Price of Stability 
in our model is $\Theta(1)$ since the social welfare is trivially bounded by $n^2$ 
and either of the two equilibrium networks above achieve a social welfare of $\Theta(n^2)$. 

Our results suggest several natural questions for future work. Our upper 
bound of $O(n \log n)$ is a logarithmic factor higher than the densest equilibrium 
network that we can create. Narrowing this gap is the most interesting open 
question. Improving our network density upper bound to $O(n)$ edges 
would immediately imply that all non-trivial equilibrium networks achieve 
$\Omega(n^2)$ social welfare. Another direction for future work is to analyze 
how network density and social welfare evolves when agents additionally 
have an option to invest in immunization that protects them from infections. 
\bibliographystyle{named}
\bibliography{bib}
\appendix
\section{Galton-Watson Branching Process}
\label{sec:useful-lem}

The Galton-Watson branching process was introduced by Galton as a mathematical model
for the propagation of family names. In the process, a population of individuals 
(e.g. people) evolve over discrete time $n=1, 2, \ldots$. Each $n$th generation individual
(i.e. individuals who are produced at time $n$) produce a random number of individuals 
(called \emph{offsprings}) independently according to some
distribution $\xi$ (called the \emph{offspring distribution}) for the $n+1$th generation. 
The goal is to study the number of individuals in the future generations. 

The process can go to extinction when after $n$ generation, with high probability, 
the number of individuals in generation $n+1$ is 0. This can happen for example when
we start from 1 individual and $\E[\xi]<1$. In case the process goes
to extinction, we are interested
in characterizing how fast this happens or how 
many individuals are generated before extinction. 

\begin{lemma}[Galton-Watson process~\cite{DraiefM09}]
\label{lem:gw}
Let $\xi$ denote the offspring distribution of each individual in the Galton-Watson process
when starting with one individual. Furthermore,  let $h=\text{sup}_{\theta\ge 0}\{\theta-\log \E[e^{\theta\xi}]\}$. 
Let $\mathcal{T}$ denote the set of of total individuals created by process 
when the process goes to extinction. Then
$\Pr[\mathcal{|T|}>k] \le e^{-kh}$ for all $k\in \mathbb{N}$.
\end{lemma}

\begin{cor} \label{cor:gw}
Let $\epsilon > 0$ and $n\in\N$. 
In the Galton-Watson process, suppose the offspring distribution $\xi$ 
is the sum of $m=O(n^{1-\epsilon})$ Bernoulli random variables with probability at most $1/n$.
Then with probability at least $1-o(n^{-2})$, the number of individuals created by the process 
is at most $3/\epsilon$.
\end{cor}
\begin{proof}
\begin{align*}
\E[e^{\theta\xi}] & = \sum_{i=1}^m \Pr[\xi=i] e^{\theta i}
\le \sum_{i=1}^m \frac{1}{n^i} \tbinom{m}{i} e^{\theta i} 
\le \sum_{i=1}^m \frac{m^i e^{\theta i}}{n^ii!}\\
&\le \sum_{i=1}^m (\frac{me^{\theta}}{n})^i 
\end{align*}
Let $\theta'=\epsilon \log n - 1$, 
$$\E[e^{\theta' \xi}] \le \sum_{i=1}^m (\frac{1}{e})^i \le \frac{e}{e-1} $$
So \begin{align*}h&=\text{sup}_{\theta\ge 0}\{\theta-\log \E[e^{\theta\xi}]\}\geq \theta'-\E[e^{\theta' \xi}] \\&\geq \epsilon\log n -1 - \frac{1}{e-1} > \epsilon \log n -3.\end{align*}
Therefore by Lemma~\ref{lem:gw}, $$\Pr[|\mathcal{T}|>\frac{3}{\epsilon}]\le \frac{e^{9/\epsilon}}{n^3} = o(n^{-2}),$$
as claimed.
\end{proof}

\end{document}